%% file: main.tex
\documentclass[12pt,a4paper]{article}

\usepackage{amsmath,amsthm,amssymb,amscd,a4wide}

\usepackage{dsfont}
\usepackage{mathrsfs}
\usepackage{setspace}
\usepackage{hyperref}

\input{def.tex}

\numberwithin{equation}{section}

\begin{document}

\thispagestyle{empty}

\vspace*{1cm}

\begin{center}

{\LARGE\bf On pairs of interacting electrons in a quantum wire} \\

\vspace*{2cm}

{\large Joachim Kerner \footnote{E-mail address: {\tt Joachim.Kerner@fernuni-hagen.de}} }%

\vspace*{5mm}

Department of Mathematics and Computer Science\\
FernUniversit\"{a}t in Hagen\\
58084 Hagen\\
Germany\\

\end{center}

\vfill

\begin{abstract} In this paper we consider pairs of interacting electrons moving in a simple quantum wire, namely the half-line $\mathbb{R}_+$. In particular, we extend the results obtained in \cite{KernerElectronPairs} by allowing for contact interactions of the Lieb-Liniger type between the two electrons constituting the pair. We construct the associated Hamiltonian rigorously and study its spectral properties. We then investigate Bose-Einstein condensation of pairs and prove, as a main result, the existence of condensation whenever the Hamiltonian has a non-trivial discrete spectrum. Most importantly, condensation is proved for very weak and very strong contact interactions. 
\end{abstract}

\newpage

\input{intro}
\input{1sec}
\vspace*{0.5cm}


\vspace*{0.5cm}

{\small
\bibliographystyle{amsalpha}
\bibliography{Literature}}

\end{document}

%% file: def.tex
\newcommand{\ud}{\mathrm{d}}

\newcommand{\cD}{{\mathcal D}}

\numberwithin{equation}{section}

\newtheorem{theorem}{Theorem}[section]
\newtheorem{lemma}[theorem]{Lemma}
\newtheorem{prop}[theorem]{Proposition}
\newtheorem{cor}[theorem]{Corollary}
\newtheorem{remark}[theorem]{Remark}

\theoremstyle{definition}


%% file: intro.tex
\section{Introduction}
As understood by Cooper \cite{CooperBoundElectron} and later by Bardeen and Schrieffer \cite{BCSI}, the pairing of electrons is the key mechanism in the formation of the (type-I) superconducting phase in metals. Most importantly, although each electron is of course a fermion, the (Cooper) pair as a whole can be treated as a bosonic particle \cite{MR04} and, as a consequence, a gas of pairs can undergo Bose-Einstein condensation \cite{EinsteinBEC}. This condensation of pairs then manifests the superconducting phase and is ultimately responsible for to the coherence in the many-electron system which is characteristic for this phase. 

This paper is the third in a series of papers~\cite{KernerElectronPairs,KernerSurfaceDefects} in which the condensation of pairs has been investigated for a simple quantum wire, i.e., the half-line $\mathbb{R}_+$. Aside from [Section~3.1,\cite{KernerElectronPairs}] where spatially localised (and hence non-separable) two-particle interactions were considered, the electrons constituting a pair have been assumed not to interact with each other except, of course, for the (binding-) interaction responsible for the pairing. However, from a physical point of view one expects the two electrons to repel each when being close. For this reason we will allow in this paper for additional (repulsive) contact interactions between the two electrons of the Lieb-Liniger type \cite{LL63}, i.e., the interaction term shall formally be given by
\begin{equation}\label{InteractionTerm}
\alpha \delta(x-y)\ ,
\end{equation}
with $\alpha \geq 0$ being the interaction strength and $\delta$ the Dirac-delta distribution. Contact interactions of this type are frequently investigated in mathematical physics \cite{Ya67,AGHH88,Ha07,Ha08,BKContact,BolteGarforthTwo} since they often provide solvable interacting many-particle models and, due to technological advances in the last decades, they have also become increasingly important in more applied areas of physics in recent years \cite{cazalilla2011one,OlshaniiDunjko,Volosniev}. Most importantly, however, contact interactions of the form \eqref{InteractionTerm} are assumed to be a good approximation of more general short-range contact interactions.

The paper is organised as follows: In the Section~\ref{Model} we introduce the formal Hamiltonian of a single pair and establish a rigorous realisation thereof. In Section~\ref{Spectrum} we perform a spectral analysis of this Hamiltonian, characterising the essential as well as the discrete part of the spectrum. Most importantly, we show that the discrete part of the spectrum is non-trivial for very weak and very strong values of the interaction strength. Finally, we use the obtained knowledge about the spectrum to investigate Bose-Einstein condensation of interacting pairs in Section~\ref{CondensationSection}. We prove that condensation exists whenever the Hamiltonian possesses eigenstates below the essential spectrum. In particular, we prove that condensation exists for very weak and very strong contact interactions.

%% file: 1sec.tex
\section{The model}
\label{Model}
We consider two electrons with opposite spin moving on a simple quantum wire, i.e., the half-line $\mathbb{R}_+=[0,\infty)$. The Hamiltonian of the system shall (formally) be given by
\begin{equation}\label{HamiltonianInteractingPair}
H_{\alpha}=-\frac{\partial^2}{\partial x^2}-\frac{\partial^2}{\partial y^2}+v_{b}(|x-y|)+\alpha \delta(x-y)\ ,
\end{equation}
where $v_b:\mathbb{R}_+ \rightarrow \overline{\mathbb{R}}$ is the binding potential 
\begin{equation}
v_b(x):=\begin{cases}
0 \quad \text{if}\quad x\leq d\ ,\\
\infty \quad \text{else}\ ,
\end{cases}
\end{equation}
leading to a bound pair of electrons (``Cooper pair'') with $d > 0$ characterising the size of the pair. 

In order to arrive at a rigorous realisation of the Hamiltonian~\eqref{HamiltonianInteractingPair} one introduces the quadratic form
\begin{equation}\label{form}
q_{\alpha}[\varphi]=\int_{\Omega}|\nabla \varphi|^2\ \ud x+\alpha \int_{\mathbb{R}}|\varphi(x,x)|^2 \ \ud x
\end{equation}
on $\cD_{q}:=\{\varphi \in H^1(\Omega): \varphi(x,y)=\varphi(y,x)\ \text{and}\ \varphi|_{\partial \Omega_D}=0 \}$ with
\begin{equation}
\Omega:=\{(x,y)\in \mathbb{R}^2_+:\ |x-y| \leq d \}
\end{equation}
and $\partial \Omega_D:=\{(x,y)\in \Omega:\ |x-y|=d \}$. Setting
\begin{equation}
L^2_s(\Omega):=\{\varphi \in L^2(\Omega):\ \varphi(x,y)=\varphi(y,x)\}
\end{equation}
we have the following statement.
\begin{theorem} On the Hilbert space $L^2_s(\Omega)$, the form~\eqref{form} is densely defined, closed and bounded from below.
\end{theorem}
\begin{proof} Density follows directly from the fact that $C^{\infty}_{0}(\tilde{\Omega}) \subset L^2(\tilde{\Omega})$ is a dense subset for $\tilde{\Omega}:=\{(x,y)\in \mathbb{R}^2_+:\ y \leq x \ \text{and}\ |x-y| \leq d\}$. Also, since the form is positive, it is bounded from below.
	
	Finally, similar to \cite{BKContact,KM16} one can establish the trace estimate
	\begin{equation}
	\int_{\mathbb{R}}|\varphi(x,x)|^2 \ \ud x  \leq c \|\varphi\|^2_{H^1(\Omega)}\ ,
\end{equation}
$c > 0$ being some constant, from which closedness readily follows.
\end{proof}
Hence, according to the representation theorem of quadratic forms~\cite{BEH08} there exists a unique self-adjoint operator being associated with $q_{\alpha}[\cdot]$ which shall be denoted by $H_{\alpha}$.
\begin{remark} The operator $H_{\alpha}$ acts as the two-dimensional Laplacian subjected to Dirichlet boundary conditions along $\partial \Omega_D$, Neumann boundary conditions along $\partial \Omega \setminus \partial \Omega_D$ and Robin boundary conditions along the diagonal $x=y$, see \cite{BKContact,KernerMuhlenbruch2}.
\end{remark}
\section{On the spectrum of $H_{\alpha}$}
\label{Spectrum}
In this section we are concerned with the spectrum of the operator $H_{\alpha}$ and in a first result we characterise the essential part thereof.

 In order to do this we consider the (self-adjoint) one-dimensional Laplacian $-\frac{\ud^2}{\ud x^2}$ on the interval $[0,\frac{d}{\sqrt{2}}]$ with operator domain
\begin{equation}\label{OperatorDomainOneDimensionalII}
\cD_{D,\alpha}:=\left\{\varphi \in H^2(0,d/\sqrt{2}):\ \varphi^{\prime}(0)-\frac{\alpha}{2\sqrt{2}}\varphi(0)=0 \quad \text{and} \quad \varphi\left(\frac{d}{\sqrt{2}}\right)=0 \right\}\ .
\end{equation}
From~\eqref{OperatorDomainOneDimensionalII} we see that one imposes Robin boundary conditions at $x=0$ and Dirichlet boundary conditions at $x=\frac{d}{\sqrt{2}}$ (regarding the choice of the constant in the Robin boundary conditions see, for example, \cite{BKContact}). This operator has purely discrete spectrum and we shall denote its eigenvalues by $\{\varepsilon^{D}_n(\alpha)\}_{n \in \mathbb{N}_0}$.
\begin{theorem}[Essential spectrum] One has
	\begin{equation}
	\sigma_{ess}(H_{\alpha})=\left[\varepsilon^{D}_0(\alpha),\infty\right)\ .
	\end{equation}
\end{theorem}
\begin{proof} A detailed proof is obtained using the methods employed in the proofs of~[Theorem~2.1,\cite{KernerElectronPairs}] and [Theorem~3.1,\cite{KernerMuhlenbruch2}]. 
	
	A sketch of the proof is as follows: Since the essential spectrum is determined by the behaviour of the domain at infinity, $\Omega$ in this sense reduces to a half-infinite rectangle with the corresponding boundary conditions. Then, using separation of variables, with one operator being $-\frac{\ud^2}{\ud x^2}$ on $\cD_{D,\alpha}$ and the other one $-\frac{\ud^2}{\ud y^2}$ on 
	\begin{equation}
	\tilde{\cD}:=\{\varphi \in H^2(0,\infty): \varphi^{\prime}(0)=0\}\ ,
	\end{equation}
	the statement follows readily since $\sigma\left(-\frac{\ud^2}{\ud y^2}\right)=[0,\infty)$.
\end{proof}
We now turn attention to the discrete part of the spectrum. We will prove that it is non-trivial given the interaction between the two particles is very weak or very strong.
\begin{theorem}[Discrete spectrum]\label{DiscreteSpectrum} There exist constants $\alpha_1,\alpha_2 > 0$ such that 
	\begin{equation}
	\sigma_{d}(H_{\alpha})\neq \emptyset
	\end{equation}
	for all $\alpha \geq 0$ such that $\alpha < \alpha_1$ or $\alpha > \alpha_2$.
\end{theorem}
\begin{proof} We first note that the discrete spectrum is non-empty for $\alpha=0$, see~\cite{KernerMuhlenbruch2}. Let $\varphi_{0} \in H^1(\Omega)$ be the corresponding (normalised) ground state with 
	\begin{equation}
	\int_{\Omega}|\nabla \varphi_0|^2\ \ud x =E_0 < \varepsilon^{D}_{0}(0)=\frac{\pi^2}{2d^2}\ .
	\end{equation}
 Since $\varepsilon^{D}_0(\alpha) \rightarrow \varepsilon^{D}_0(0)$ as $\alpha \rightarrow 0$, there exists a constant $\alpha_1 > 0$ such that 
	\begin{equation}
	E_0+\alpha \int_{\mathbb{R}}|\varphi_0(x,x)|^2 \ \ud x <\varepsilon^{D}_{0}(\alpha)
\end{equation}
for all $\alpha < \alpha_1$. Hence, using $\varphi_0$ as a trial function, the statement follows by the Rayleigh-Ritz variational principle \cite{BEH08}.

Now, in \cite{KernerElectronPairs} the functions in the form domain were assumed to be anti-symmetric with respect to the diagonal $x=y$ which effectively implies Dirichlet boundary conditions (due to continuity). Dirichlet boundary conditions along the diagonal, on the other hand, correspond to the case $\alpha=\infty$ for which the associated one-dimensional Laplacian has domain
\begin{equation}
\cD_{D,\infty}:=\left\{\varphi \in H^2(0,d/\sqrt{2}):\ \varphi(0)=0 \quad \text{and} \quad \varphi\left(\frac{d}{\sqrt{2}}\right)=0 \right\}\ ,
\end{equation}
and eigenvalues $\{\varepsilon^D_n(\infty)\}_{n \in \mathbb{N}_0}$. Furthermore, as proved in~\cite{KernerElectronPairs}, the discrete spectrum of $H_{\alpha}$ is non-empty in this case, i.e., for $\alpha=\infty$.

Let $\tilde{\varphi}_0$ be the corresponding ground state with 
	\begin{equation}
	\int_{\Omega}|\nabla \tilde{\varphi}_0|^2\ \ud x =\tilde{E}_0 < \varepsilon^{D}_{0}(\infty)=\frac{2\pi^2}{d^2}\ .
	\end{equation}
Now, since $\varepsilon^{D}_0(\alpha) \rightarrow \varepsilon^{D}_0(\infty)$ as $\alpha \rightarrow \infty$, there exists a constant $\alpha_2 > 0$ such that 
\begin{equation}
q_{\alpha}[\tilde{\varphi}_0] < \varepsilon^{D}_{0}(\alpha)
\end{equation}
which proves the statement. Note here that $\tilde{\varphi}_0$ has vanishing trace along $x=y$.
\end{proof}
\begin{remark} Theorem~\ref{DiscreteSpectrum} has an important physical consequence. Namely, even for very large interaction strengths $\alpha$, there exists at least on bound state. This means, intuitively speaking, that strong contact interactions do not destabilise completely. As shown in [Lemma~3.5,\cite{KernerElectronPairs}], this is not the case for other types of singular two-particle interactions which are not contact interactions. 
\end{remark}
In a next step consider the (self-adjoint) one-dimensional Laplacian $-\frac{\ud^2}{\ud x^2}$ on the interval $[0,\frac{d}{\sqrt{2}}]$ with operator domain
\begin{equation}\label{OperatorDomainOneDimensional}
\cD_{N,\alpha}:=\left\{\varphi \in H^2(0,d/\sqrt{2}):\ \varphi^{\prime}(0)-\frac{\alpha}{2\sqrt{2}}\varphi(0)=0 \quad \text{and} \quad \varphi^{\prime}\left(\frac{d}{\sqrt{2}}\right)=0 \right\}\ .
\end{equation}
From~\eqref{OperatorDomainOneDimensional} we see that one imposes Robin boundary conditions at $x=0$ and Neumann boundary conditions at $x=\frac{d}{\sqrt{2}}$. Denoting its eigenvalues by $\{\varepsilon^N_n(\alpha)\}_{n \in \mathbb{N}_0}$ we arrive at the following statement.
\begin{lemma}[Ground state energy] Let $E_{0}(\alpha):=\inf \sigma(H_{\alpha})$ denote the ground state energy of $H_{\alpha}$. Then one has the estimate
	\begin{equation}
	2\varepsilon^N_0(\alpha)\leq  E_{0}(\alpha) \leq \varepsilon^{D}_0(\alpha)\ .
	\end{equation}
\end{lemma}
\begin{proof} This statement follows using the methods employed in the proof of Theorems~2.3,~2.4 of \cite{KernerElectronPairs}.
\end{proof}
Finally, we have the following statement.
\begin{prop}\label{NumberEigenvalues} For any $\alpha \geq 0$, the number of eigenvalues of $H_{\alpha}$ smaller than $\varepsilon^{D}_0(\alpha)$ is finite.
\end{prop}
\begin{proof} The statement follows by an operator-bracketing argument similar to the one used in the proof of [Theorem~2.4,\cite{KernerElectronPairs}], see also \cite{KernerMuhlenbruch2}.
\end{proof}
\section{On the condensation of pairs and the superconducting phase}
\label{CondensationSection}
The condensation phenomenon in a system of non-interacting pairs with zero interaction strength (i.e., $\alpha=0$) was investigated in~\cite{KernerElectronPairs}. In this section it is our goal to generalise the results to some cases where $\alpha > 0$.

The standard treatment of Bose-Einstein condensation in quantum statistical mechanics requires a thermodynamic limit \cite{RuelleSM}. In a first step one therefore reduces the one-pair configuration space from the half-line $\mathbb{R}_+$ to the interval $[0,L]$. Consequently, the one-pair Hilbert space is $L^2_s(\Omega_L)$ with
\begin{equation}
\Omega_L:=\{(x,y) \in \Omega: x,y \leq L \}\ .
\end{equation}
On this Hilbert space on then constructs a form $q^L_{\alpha}$ as the obvious version of $q_{\alpha}$ on $L^2_s(\Omega_L)$ which yields the self-adjoint operator $H^{L}_{\alpha}$, i.e., the version of $H_{\alpha}$ on $L^2_s(\Omega_L)$. Note that we impose Dirichlet boundary conditions along the boundary segments of $\Omega_L$ for which $x=L$ or $y=L$, see also \cite{KernerElectronPairs} for more details.

Since $\Omega_L$ is a bounded Lipschitz domain, the spectrum of $H^{L}_{\alpha}$ is purely discrete. Let $\{E^L_n(\alpha)\}_{n \in \mathbb{N}_0}$ be the corresponding eigenvalues, counted with multiplicity. In a first result we establish a finite-volume version of Proposition~\ref{NumberEigenvalues}, being proved similarly.
\begin{prop} The number of eigenvalues of $H^L_{\alpha}$ smaller than $\varepsilon^{D}_0(\alpha)$ is uniformly bounded for all $L > d$.
\end{prop}
Working in the grand-canonical ensemble, the number of pairs occupying the eigenstate with eigenvalue $E^L_n(\alpha)$ at inverse temperature $\beta=\frac{1}{T}$ is given by
\begin{equation}
\frac{1}{e^{\beta(E^L_n(\alpha)-\mu_L)}}\ ,
\end{equation}
$\mu_L \leq E^L_0(\alpha)$ denoting the chemical potential \cite{RuelleSM}. The thermodynamic limit is then defined as the limit $L \rightarrow \infty$ such that 
\begin{equation}\label{ChemicalPotentialCondition}
\frac{1}{L}\sum_{n=0}^{\infty}\frac{1}{e^{\beta(E^L_n(\alpha)-\mu_L)}}=\rho
\end{equation}
holds for all values of $L$ with $\rho > 0$ being the density of pairs. 

Now, we obtain the following result which is obtained in analogy to [eq.~(3.4),\cite{KernerElectronPairs}].
\begin{prop}\label{ParticleNumberExcitedStates} Let $\mu < \varepsilon^{D}_0(\alpha)$ be given. Then
	\begin{equation}
	\lim_{L \rightarrow \infty}\frac{1}{L}\sum_{n:E^L_n(\alpha) \geq \varepsilon^{D}_0(\alpha) }\frac{1}{e^{\beta(E^L_n(\alpha)-\mu)}}=\frac{1}{\pi}\sum_{n=0}^{\infty}\int_{0}^{\infty}\frac{1}{e^{\beta \varepsilon^{D}_n(\alpha)}e^{\beta(x^2-\mu)}-1}\ \ud x\ .
	\end{equation}\ .
\end{prop}
\begin{remark} Proposition~\ref{ParticleNumberExcitedStates} characterises the density of pairs occupying the states with eigenvalues not smaller than the bottom of the essential spectrum of $H_{\alpha}$.
\end{remark}
With Proposition~\ref{ParticleNumberExcitedStates} at hand we can now prove the main result of this section. For this note that we say an eigenstate $\varphi^{L}_n(\alpha) \in H^1(\Omega_L)$ with associated eigenvalue $E^L_n(\alpha)$ is macroscopically occupied in the thermodynamic limit if 
\begin{equation}
\limsup_{L \rightarrow \infty}\frac{1}{e^{\beta(E^L_n(\alpha)-\mu_L)}} > 0\ .
\end{equation}
\begin{theorem}[Condensation of interacting pairs] Let $\alpha \geq 0$ be such that $H_{\alpha}$ has non-trivial discrete spectrum. Then the ground state $\varphi^{L}_0(\alpha)$ is macroscopically occupied in the thermodynamic limit. 
\end{theorem}
\begin{proof} Based on [Lemma~3.1,\cite{KernerElectronPairs}] and [Proposition~3.2,\cite{KernerElectronPairs}] we first observe that the assumptions imply the existence of a value $L_0$ such that $H^{L}_{\alpha}$ has an eigenvalue smaller than $\varepsilon^{D}_0(\alpha)$ for all values $L > L_0$. Furthermore, the lowest eigenvalue of $H^{L}_{\alpha}$ converges to $\inf \sigma(H_{\alpha}) < \varepsilon^{D}_0(\alpha)$ as $L \rightarrow \infty$.
	
	Now, since one has $\mu_L < \inf \sigma(H^L_{\alpha})$ for all $L$, we conclude that $\mu_L < \inf \sigma(H_{\alpha})+\varepsilon_1$ for all $L$ large enough and $\varepsilon_1 > 0$ arbitrarily small. Accordingly we have $\mu_L < \varepsilon^{D}_0(\alpha)-\varepsilon_2:=\mu$ for some small $\varepsilon_2 > 0$ and $L$ large enough. This then allows us to arrive at the estimate 
	\begin{equation}\label{EstimateCondensation}\begin{split}
	\frac{1}{L}\sum_{n:E^L_n(\alpha) \geq \varepsilon^{D}_0(\alpha) }\frac{1}{e^{\beta(E^L_n(\alpha)-\mu_L)}}&\leq \frac{1}{L}\sum_{n:E^L_n(\alpha) \geq \varepsilon^{D}_0(\alpha) }\frac{1}{e^{\beta(E^L_n(\alpha)-\mu)}}\\
	&\leq \frac{1}{\pi}\sum_{n=0}^{\infty}\int_{0}^{\infty}\frac{1}{e^{\beta \varepsilon^{D}_n(\alpha)}e^{\beta(x^2-\mu)}-1}\ud x +\varepsilon_3(L)\ ,
	\end{split}
	\end{equation}
	for some $\varepsilon_3(L) > 0$ with $\varepsilon_3(L) \rightarrow 0$ as $L\rightarrow \infty$, taking Proposition~\ref{ParticleNumberExcitedStates} into account. 
	
	The important fact is now that $\varepsilon_3,\mu$ are independent of the pair density $\rho > 0$ which affects only the sequence $(\mu_L)$ according to~\eqref{ChemicalPotentialCondition}. Hence, comparing \eqref{EstimateCondensation} with \eqref{ChemicalPotentialCondition} one concludes that 
	\begin{equation}
	\limsup_{L \rightarrow \infty}\frac{1}{L}\sum_{n:E^L_n(\alpha) < \varepsilon^{D}_0(\alpha) }\frac{1}{e^{\beta(E^L_n(\alpha)-\mu_L)}} > 0
	\end{equation}
	for a large enough pair density $\rho > 0$. From this the statement follows immediately, since the ground state is the state which is occupied the most.
\end{proof}
Applying Theorem~\ref{DiscreteSpectrum}, we readily obtain the following result.
\begin{cor}\label{CorollarySuperPhase} Let $\alpha$ be as characterised in Theorem~\ref{DiscreteSpectrum}. Then the ground state of $H^L_{\alpha}$ is macroscopically occupied in the thermodynamic limit.
\end{cor}
\begin{remark} As described in \cite{KernerElectronPairs,KernerSurfaceDefects}, a macroscopic occupation of a single-pair state in the thermodynamic limit is associated with a superconducting phase in the bulk~\cite{BCSI,CooperBoundElectron,MR04}. Hence, Corollary~\ref{CorollarySuperPhase} shows that even very strong contact interactions of the Lieb-Liniger type do not lead to a destruction of the superconducting phase (in the bulk). Note that this is in sharp contrast to [Remark~3.7,\cite{KernerElectronPairs}] in which the effect of non-separable singular two-particle interactions is discussed.
\end{remark}